\documentclass[11pt,letterpaper]{article}
\usepackage[margin=1in]{geometry}
\usepackage{graphicx}
\usepackage{amssymb,amsthm,amsmath}
\usepackage{lmodern}
\usepackage[T1]{fontenc}
\usepackage{algorithm}
\usepackage[hidelinks]{hyperref}
\usepackage[indLines=true]{algpseudocodex}
\usepackage{tikz}
\usetikzlibrary{positioning}
\usetikzlibrary{arrows.meta, quotes}
\usetikzlibrary{decorations.pathmorphing,calligraphy}

\newtheorem{theorem}{Theorem}

\newtheorem{lemma}{Lemma}
\newtheorem{corollary}{Corollary}

\theoremstyle{definition}
\newtheorem{definition}{Definition}
\newtheorem{remark}{Remark}

\theoremstyle{remark}

\title{%
	\bfseries\boldmath Space-Efficient Hierholzer:\\ Eulerian Cycles in $\mathrm{O}(m)$ Time and $\mathrm{O}(n)$ Space%
	\thanks{%
		Ziad Ismaili Alaoui and Sebastian Wild are supported by EPSRC grant EP/X039447/1.%
	}
}

\author{
  Ziad Ismaili Alaoui\thanks{University of Liverpool, United Kingdom. \texttt{\{ziad.ismaili-alaoui}, \texttt{sebastian.wild\}\,@\,liverpool.ac.uk}} 
  \and 
  Detlef Plump\thanks{University of York, United Kingdom. \texttt{detlef.plump\,@\,york.ac.uk}}
  \and 
  Sebastian Wild\footnotemark[1] \thanks{Philipps-Universität Marburg, Germany. \texttt{wild\,@\,informatik.uni-marburg.de}}
}

\date{}

\begin{document}

\maketitle

\begin{abstract}
	We describe a simple variant of Hierholzer's algorithm that finds an Eulerian cycle in a directed (multi)graph with $n$ vertices and $m$ edges using $\mathrm{O}(n \lg m)$ bits of working memory. 
	This substantially improves the working space compared to standard implementations of Hierholzer's algorithm, which use $\mathrm{O}(m \lg n)$ bits of space. 
	Our algorithm runs in linear time, like the classical versions, but avoids an $\mathrm{O}(m)$-size stack of vertices or storing information for each edge. 
	To our knowledge, this is the first linear-time algorithm to achieve this space bound, 
	and the method is very easy to implement.
	The correctness argument, by contrast, is surprisingly subtle; we give a detailed formal proof.
    The space savings are particularly relevant for dense graphs or multigraphs with large edge multiplicities.
    \\[1ex]
    \textsl{\footnotesize An interactive visualization of the algorithm is available at \href{https://www.wild-inter.net/publications/hierholzer}{\ttfamily wild-inter.net/publications/hierholzer}.}
\end{abstract}

\section{Introduction.}
\label{s:introduction}
An \emph{Eulerian cycle} in a graph is a closed walk that traverses every edge exactly once. A~graph is called \emph{Eulerian} if it contains an Eulerian cycle.
Euler's study of the existence of Eulerian cycles and their computation is often cited as the birth of modern graph theory.
Apart from historical significance, efficiently computing them has applications in, e.g., genome assembly~\cite{pevzner2001eulerian}, routing problems~\cite{orloff1974fundamental}, and 3D printing~\cite{yamamoto2022novel}. 

Classic algorithms such as Fleury's algorithm or Hierholzer's algorithm~\cite{hierholzer1873moglichkeit}, both dating from the late 19th century, are simple to describe and find an Eulerian cycle in any Eulerian graph in polynomial time. 
Indeed, both consider each edge only once, and Hierholzer's algorithm can be implemented to run in optimal linear time.
Being both simple and efficient, the latter is the method of choice in practice.
Typical implementations of Hierholzer's algorithm (see below) use $\mathrm{O}(m \lg n)$ bits of working memory. (Here and throughout, $\lg=\log_2$.)

In this paper, we present a simple variant of Hierholzer's algorithm that uses only $\mathrm{O}(n \lg m)$ bits of working memory and runs in optimal $\mathrm{O}(n+m)$ time.  For dense graphs or multigraphs with many parallel edges, this is a substantial saving.
Note that this space can neither hold the entire input nor the entire output of the problem;
we treat the input graph as given in read-only memory and produce the Eulerian cycle \emph{in order} to a write-only stream.
In particular, our algorithm would be suitable for an application that directly consumes and uses edges of the Eulerian cycle as they are computed, potentially without ever storing the entire cycle.
By \emph{working memory}, we mean the extra space occupied by auxiliary data structures during the computation.
Our algorithm relies on an assumption about the graph representation, which is satisfied in typical adjacency-list implementations: for directed graphs, we require that we can iterate over incoming and outgoing edges; for undirected graphs, we (instead) require a consistent ordering of vertices across all adjacency lists.

To our knowledge, we present the first algorithm with working space $\mathrm{O}(n \lg m)$ bits that runs in linear time and explicitly produces an Eulerian cycle as its output.
Our algorithm is also usable in rule-based models of computation, such as graph-transformation languages like GP\,2~\cite{plump2012design}, which do not natively support stacks or recursion.
This is not known to be true for other standard efficient implementations of Eulerian cycle algorithms.

\medskip

\noindent
\textsl{\textbf{Update of August 2026:} 
Unlike hinted at in a previous version of this paper (including the version at SOSA 2026~\cite{IsmailiAlaouiPlumpWild2026}),
the algorithm presented in this work
does not easily extend to a space-efficient algorithm for undirected graphs.
A rather different approach (Grigorescu, Ismaili Alaoui, Nakajima, Shirazi Mofrad, and Wild~\cite{grigorescu2026space}), albeit still based on
Hierholzer's algorithm, does achieve the same time and space
guarantees on undirected graphs, but the techniques are different. 
That latter approach provably does not work on directed graphs; hence our result is orthogonal.}

\subsection{Hierholzer's Algorithm.}

Hierholzer's algorithm was originally described in 1873~\cite{hierholzer1873moglichkeit}, 
unsurprisingly without detailed information about data structures for efficient execution on a computer.
The original algorithm consists of following an arbitrary walk in an Eulerian graph until we get stuck, which can only happen when we close a cycle.  If this cycle has not used all edges yet, starting at such an unused edge, again following an arbitrary walk of unused edges can again only get stuck when closing a cycle; we can then fuse this new cycle into the previous one as a detour before continuing the original cycle.
By iterating this procedure, we eventually obtain an Eulerian cycle.\footnote{Throughout the paper, we always assume the graphs to be Eulerian, i.e. to contain an Eulerian cycle.}

\subsection{Typical Implementations.}

Several implementations of Hierholzer's algorithm are in use.
Closest to the above description is using a linked list of edges to represent the Eulerian cycle;
then one can insert cyclic detours as sketched above at any point in the tour.
This approach is used, e.g., in the widely used Java graph library \emph{JGraphT}~\cite{jgrapht}.
It clearly needs $\Theta(m \lg m)$ bits of working memory for storing the linked list of edges.

More amenable to producing the output directly in the correct order is the approach to treat Hierholzer's algorithm as an \emph{edge-centric} depth-first search. 
We follow the basic steps of a graph traversal, but mark edges rather than vertices. So we only backtrack from a vertex when all of its incident edges have been traversed.
When this happens for the first time, the last edge traversed to this dead end was the edge closing the first cycle of Hierholzer's algorithm, and can thus be output now as the \emph{last} edge of the Eulerian cycle. Iteratively, the same is true for the next dead end vertex.\footnote{
	See, e.g., slide 33 of these lecture notes: \url{https://www.wild-inter.net/teaching/ea/09-graph-algorithms.pdf\#page=40}.
}
Both the stack and the marking of edges require $\Theta(m \lg m)$ bits of working space in the worst case.
A minor twist of this approach, where edges are not marked, but removed from a copy of the input graph,
is used, e.g., in \emph{NetworkX}~\cite{networkx}, a popular graph framework for Python.
The copy, of course, also uses $\Omega(m \lg n)$ bits of space.

\subsection{Our Idea.}

Our key idea is the following observation. 
While an edge-centric DFS is \emph{one} correct implementation of Hierholzer's algorithm,
it is a needlessly specific one.
Concretely, there is no need to keep the entire constructed walk on a stack;
the order in which we add remaining ``detour cycles'' to the tour is immaterial for Eulerian cycles anyway.
The only constraint we have is that we \emph{reserve} the edges of the \emph{original} cycle/traversal to be the last ones to backtrack on, to make sure that we end at the starting vertex and do not omit any edges.
For that, it suffices to remember one incoming edge per vertex.
The second thing we need to remember is which edges have already been traversed;
but for that, we can store for each vertex an iterator in its adjacency list.

\subsection{Directed and Undirected Eulerian Cycles.}

In what follows, we consider \emph{directed} Eulerian multigraphs that are given as an adjacency-list representation
for outgoing \emph{and} incoming edges.
It is convenient to (logically) reverse the directions of edges; then the order of edges produced via backtracking from the edge-centric DFS directly gives the Eulerian cycle in the correct order.
(By our assumption, this is only a question of interpretation in the algorithm.)

We point out that Hierholzer's algorithm in principle works equally well for directed and undirected graphs, and it can deal with parallel edges even if these are represented as duplicate entries in the adjacency list.
Our conceptual algorithm also works for both directed and undirected graphs, however it is unclear 
how to implement it on \emph{undirected} graphs to use only $O(n\lg m)$ bits of working space.

\subsection{Related Work.}

Glazik, Schiemann, and Srivastav~\cite{glazik2023one} present a streaming algorithm for Eulerian cycles.  While they also achieve $\mathrm{O}(n\lg n)$-space and in a much more restrictive model, it is not clear if their algorithm can be implemented to run in linear total time, and they construct the cycle implicitly via a successor function.
Their algorithm is arguably more complicated than Hierholzer's algorithm.

Hagerup, Kammer, and Laudahn~\cite{hagerup2019space} proposed an algorithm using $\mathrm{O}(n+m)$ time and $\mathrm{O}(n+m)$ \emph{bits} of space for computing Eulerian cycles. They do not need strong assumptions regarding the graph data representation, but focus on undirected graphs only.  They build a nontrivial dynamic data structure in each vertex that stores which pairs of edges are adjacent in the Eulerian cycle.  Their algorithm is considerably more involved than the typical implementations of Hierholzer's algorithm using $\mathrm{O}(n+m)$ \emph{words} of space.


\subsection{Outline.}

Our paper is organised as follows. Section~\ref{s:preliminaries} introduces basic notation. Section~\ref{s:algorithm} presents the \textsc{Space-Efficient-Hierholzer} algorithm. Finally, Section~\ref{s:analysis} proves correctness, establishes the $\mathrm{O}(m)$ running time, and shows that the working memory usage is $\mathrm{O}(n \lg m)$.

\section{Preliminaries.}
\label{s:preliminaries}
We write $[n]$ for $\{1, 2, \dots, n\}$. Throughout the paper, $G = (V, E)$ denotes a finite directed multigraph, where $V$ is the set of vertices and $E$ is the multiset of directed edges, with $|V| = n$ and $|E| = m$. Loop edges are allowed. We assume without loss of generality that $V = [n]$. We denote by $d^+(v)$ and $d^-(v)$ the out-degree and in-degree of vertex~$v$, respectively. The $i$th out-neighbour of a vertex $v$ is denoted by $\Gamma^+(v, i)$, and its $i$th in-neighbour by $\Gamma^-(v, i)$, assuming arbitrary but fixed orderings of incoming and outgoing edges for each vertex.
We concisely write $uv$ for a (directed) edge from $u$ to $v$. 

\begin{definition}[Eulerian, Eulerian cycle]
    A directed graph $G$ is \emph{Eulerian} if and only if it is strongly connected\/\footnote{A set of vertices is considered \emph{strongly connected} if, for every pair of vertices $u,v$, there is a path from $u$ to $v$. This condition rules out isolated vertices.} and the in-degree equals the out-degree at every vertex; i.e. $d^+(v) = d^-(v)$ for all $v \in V$. An \emph{Eulerian cycle} in $G$ is a closed trail that traverses every edge exactly once and returns to its starting vertex.
\end{definition}

\noindent A \emph{trail} in a graph is a walk in which all edges are distinct. A \emph{closed trail} is a trail whose first and last vertices coincide.

We assume that basic arithmetic and memory operations on $\lg m$-bit integers take constant time, and that the in- and out-degree functions as well as in- and out-neighbour queries can be evaluated in constant time per operation.

\section{Space-Efficient Implementation of Hierholzer's Algorithm.}
\label{s:algorithm}
We now describe our algorithm, which computes an Eulerian cycle in a directed graph using only $\mathrm{O}(n \lg m)$ bits of working space, where $n$ is the number of vertices and $m$ is the number of edges.

\subsection{Conceptual Algorithm.}

Let $G$ be a directed Eulerian graph, and let $v_0$ be the starting vertex (it may be chosen arbitrarily). Unlike the classical approach, our algorithm traverses the graph by following \emph{incoming} edges rather than outgoing ones; that is, we perform a reverse traversal of the graph. This reversal allows the algorithm to compute the tour incrementally, in the correct order, without needing to store the entire tour and reverse it at the end.

We first describe the conceptual algorithm in terms of edge markings (colouring edges \textsc{Black}, \textsc{Red}, \textsc{Green}, or \textsc{Dashed}) to clarify the traversal and backtracking logic. 
Our implementation will represent these only implicitly. 

Initially, all edges are \textsc{Black}. 
The algorithm begins at an arbitrary starting vertex $v_0$ and follows \textsc{Black} incoming edges via a reverse traversal. When a vertex $v$ is reached for the first time in this traversal, via a \textsc{Black} edge $e=uv$, we colour $e$ \textsc{Red}. Otherwise, if we traverse an edge $e=uv$ to a vertex $v$ that had already been reached before, we colour $e$ \textsc{Green}.

\colorlet{teal}{green!50!black}
\tikzset{
	every picture/.style={xscale=1.2},
}
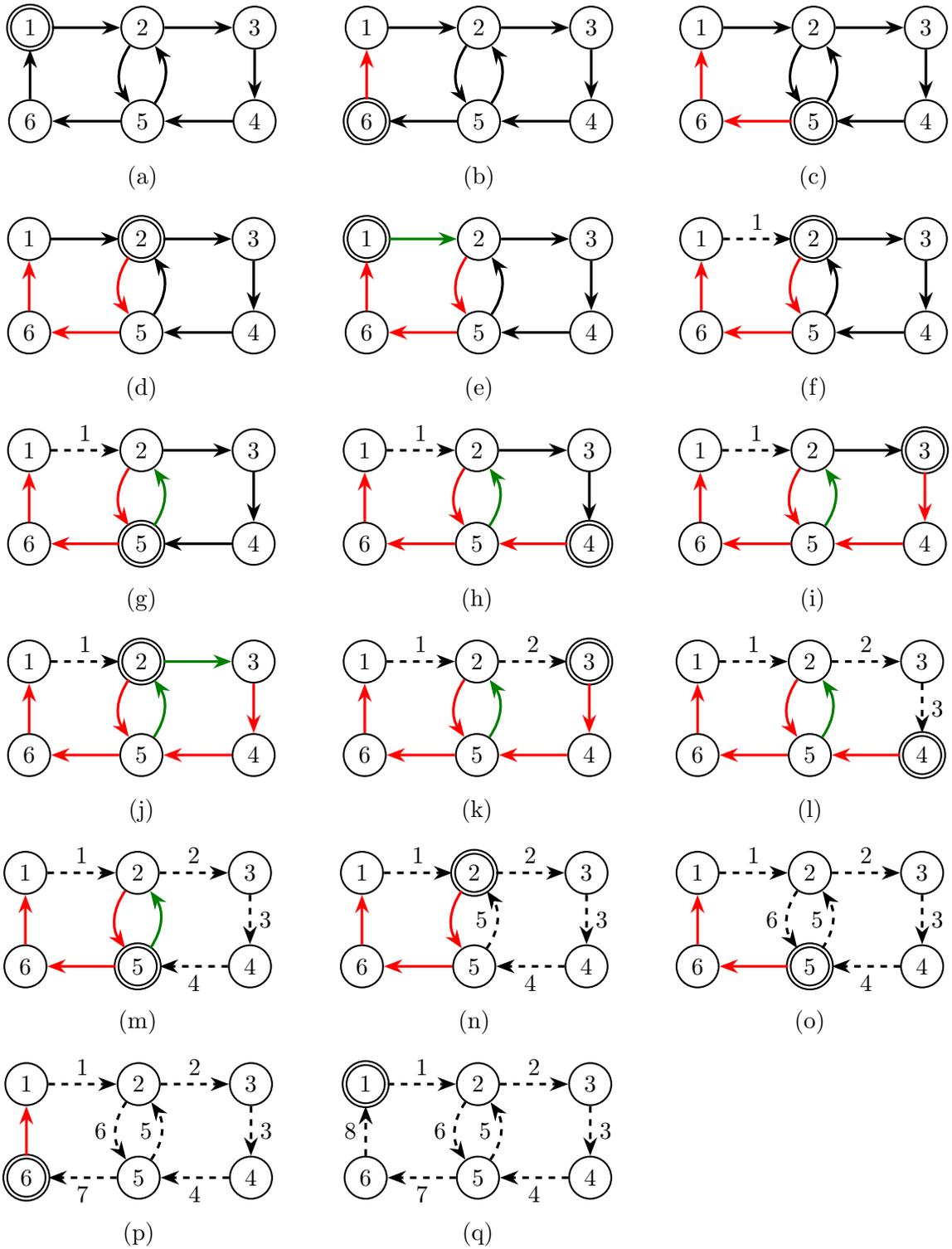
\begin{figure}[p]
    \centering
    \centering
\begin{tabular}{c@{\qquad}c@{\qquad}c}
 \begin{tikzpicture}[every edge/.style = {draw, -Stealth}]]
        \node(l) at (1.5,-2.4){(a)};
        \node[circle, double, double distance = 1pt, rounded corners=7, thick, draw=black, minimum size=5mm](1) at (0,0){$1$};
        \tikzstyle{every node}=[font=\ttfamily]
        \node[circle, rounded corners=7, thick, draw=black, minimum size=5mm](2) at (1.5,0){$2$};
        \tikzstyle{every node}=[font=\ttfamily]
        \node[circle, rounded corners=7, thick, draw=black, minimum size=5mm](3) at (3,0){$3$};
        \draw[->, thick];
        \node[circle, rounded corners=7, thick, draw=black, minimum size=5mm](4) at (3,-1.5){$4$};
        \node[circle, rounded corners=7, thick, draw=black, minimum size=5mm](5) at (1.5,-1.5){$5$};
        \node[circle, rounded corners=7, thick, draw=black, minimum size=5mm](6) at (0,-1.5){$6$};
        \draw[->, very thick] (1) edge (2) (2) edge (3) (3) edge (4) (4) edge (5) (5) edge (6) (6) edge (1) (2) edge[bend right] (5) (5) edge[bend right] (2);
    \end{tikzpicture} & \begin{tikzpicture}[every edge/.style = {draw, -Stealth}]]
        \node(l) at (1.5,-2.4){(b)};
        \node[circle, rounded corners=7, thick, draw=black, minimum size=5mm](1) at (0,0){$1$};
        \tikzstyle{every node}=[font=\ttfamily]
        \node[circle, rounded corners=7, thick, draw=black, minimum size=5mm](2) at (1.5,0){$2$};
        \tikzstyle{every node}=[font=\ttfamily]
        \node[circle, rounded corners=7, thick, draw=black, minimum size=5mm](3) at (3,0){$3$};
        \draw[->, thick];
        \node[circle, rounded corners=7, thick, draw=black, minimum size=5mm](4) at (3,-1.5){$4$};
        \node[circle, rounded corners=7, thick, draw=black, minimum size=5mm](5) at (1.5,-1.5){$5$};
        \node[circle, double, double distance = 1pt, rounded corners=7, thick, draw=black, minimum size=5mm](6) at (0,-1.5){$6$};
        \draw[->, very thick] (1) edge (2) (2) edge (3) (3) edge (4) (4) edge (5) (5) edge (6) (6) edge[red] (1) (2) edge[bend right] (5) (5) edge[bend right] (2);
    \end{tikzpicture} & \begin{tikzpicture}[every edge/.style = {draw, -Stealth}]]
        \node(l) at (1.5,-2.4){(c)};
        \node[circle, rounded corners=7, thick, draw=black, minimum size=5mm](1) at (0,0){$1$};
        \tikzstyle{every node}=[font=\ttfamily]
        \node[circle, rounded corners=7, thick, draw=black, minimum size=5mm](2) at (1.5,0){$2$};
        \tikzstyle{every node}=[font=\ttfamily]
        \node[circle, rounded corners=7, thick, draw=black, minimum size=5mm](3) at (3,0){$3$};
        \draw[->, thick];
        \node[circle, rounded corners=7, thick, draw=black, minimum size=5mm](4) at (3,-1.5){$4$};
        \node[circle, double, double distance = 1pt, rounded corners=7, thick, draw=black, minimum size=5mm](5) at (1.5,-1.5){$5$};
        \node[circle, rounded corners=7, thick, draw=black, minimum size=5mm](6) at (0,-1.5){$6$};
        \draw[->, very thick] (1) edge (2) (2) edge (3) (3) edge (4) (4) edge (5) (5) edge[red] (6) (6) edge[red] (1) (2) edge[bend right] (5) (5) edge[bend right] (2);
    \end{tikzpicture} \\ 
 \begin{tikzpicture}[every edge/.style = {draw, -Stealth}]]
        \node(l) at (1.5,-2.4){(d)};
        \node[circle, rounded corners=7, thick, draw=black, minimum size=5mm](1) at (0,0){$1$};
        \tikzstyle{every node}=[font=\ttfamily]
        \node[circle, double, double distance = 1pt, rounded corners=7, thick, draw=black, minimum size=5mm](2) at (1.5,0){$2$};
        \tikzstyle{every node}=[font=\ttfamily]
        \node[circle, rounded corners=7, thick, draw=black, minimum size=5mm](3) at (3,0){$3$};
        \draw[->, thick];
        \node[circle, rounded corners=7, thick, draw=black, minimum size=5mm](4) at (3,-1.5){$4$};
        \node[circle, rounded corners=7, thick, draw=black, minimum size=5mm](5) at (1.5,-1.5){$5$};
        \node[circle, rounded corners=7, thick, draw=black, minimum size=5mm](6) at (0,-1.5){$6$};
        \draw[->, very thick] (1) edge (2) (2) edge (3) (3) edge (4) (4) edge (5) (5) edge[red] (6) (6) edge[red] (1) (2) edge[bend right, red] (5) (5) edge[bend right] (2);
    \end{tikzpicture} & \begin{tikzpicture}[every edge/.style = {draw, -Stealth}]]
        \node(l) at (1.5,-2.4){(e)};
        \node[circle, double, double distance = 1pt, rounded corners=7, thick, draw=black, minimum size=5mm](1) at (0,0){$1$};
        \tikzstyle{every node}=[font=\ttfamily]
        \node[circle, rounded corners=7, thick, draw=black, minimum size=5mm](2) at (1.5,0){$2$};
        \tikzstyle{every node}=[font=\ttfamily]
        \node[circle, rounded corners=7, thick, draw=black, minimum size=5mm](3) at (3,0){$3$};
        \draw[->, thick];
        \node[circle, rounded corners=7, thick, draw=black, minimum size=5mm](4) at (3,-1.5){$4$};
        \node[circle, rounded corners=7, thick, draw=black, minimum size=5mm](5) at (1.5,-1.5){$5$};
        \node[circle, rounded corners=7, thick, draw=black, minimum size=5mm](6) at (0,-1.5){$6$};
        \draw[->, very thick] (1) edge[teal] (2) (2) edge (3) (3) edge (4) (4) edge (5) (5) edge[red] (6) (6) edge[red] (1) (2) edge[bend right, red] (5) (5) edge[bend right] (2);
    \end{tikzpicture} & \begin{tikzpicture}[every edge/.style = {draw, -Stealth}]]
        \node(l) at (1.5,-2.4){(f)};
        \node[circle, rounded corners=7, thick, draw=black, minimum size=5mm](1) at (0,0){$1$};
        \node[circle, double, double distance = 1pt, rounded corners=7, thick, draw=black, minimum size=5mm](2) at (1.5,0){$2$};
        \node[circle, rounded corners=7, thick, draw=black, minimum size=5mm](3) at (3,0){$3$};
        \draw[->, thick];
        \node[circle, rounded corners=7, thick, draw=black, minimum size=5mm](4) at (3,-1.5){$4$};
        \node[circle, rounded corners=7, thick, draw=black, minimum size=5mm](5) at (1.5,-1.5){$5$};
        \node[circle, rounded corners=7, thick, draw=black, minimum size=5mm](6) at (0,-1.5){$6$};
        \draw[->, very thick] (1) edge[dashed, "1"] (2) (2) edge (3) (3) edge (4) (4) edge (5) (5) edge[red] (6) (6) edge[red] (1) (2) edge[bend right, red] (5) (5) edge[bend right] (2);
    \end{tikzpicture} \\
 \begin{tikzpicture}[every edge/.style = {draw, -Stealth}]]
        \node(l) at (1.5,-2.4){(g)};
        \node[circle, rounded corners=7, thick, draw=black, minimum size=5mm](1) at (0,0){$1$};
        \node[circle, rounded corners=7, thick, draw=black, minimum size=5mm](2) at (1.5,0){$2$};
        \node[circle, rounded corners=7, thick, draw=black, minimum size=5mm](3) at (3,0){$3$};
        \draw[->, thick];
        \node[circle, rounded corners=7, thick, draw=black, minimum size=5mm](4) at (3,-1.5){$4$};
        \node[circle, rounded corners=7, double, double distance = 1pt,  thick, draw=black, minimum size=5mm](5) at (1.5,-1.5){$5$};
        \node[circle, rounded corners=7, thick, draw=black, minimum size=5mm](6) at (0,-1.5){$6$};
        \draw[->, very thick] (1) edge[dashed, "1"] (2) (2) edge (3) (3) edge (4) (4) edge (5) (5) edge[red] (6) (6) edge[red] (1) (2) edge[bend right, red] (5) (5) edge[bend right, teal] (2);
    \end{tikzpicture} & \begin{tikzpicture}[every edge/.style = {draw, -Stealth}]]
        \node(l) at (1.5,-2.4){(h)};
        \node[circle, rounded corners=7, thick, draw=black, minimum size=5mm](1) at (0,0){$1$};
        \node[circle, rounded corners=7, thick, draw=black, minimum size=5mm](2) at (1.5,0){$2$};
        \node[circle, rounded corners=7, thick, draw=black, minimum size=5mm](3) at (3,0){$3$};
        \draw[->, thick];
        \node[circle, rounded corners=7, double, double distance = 1pt, thick, draw=black, minimum size=5mm](4) at (3,-1.5){$4$};
        \node[circle, rounded corners=7, thick, draw=black, minimum size=5mm](5) at (1.5,-1.5){$5$};
        \node[circle, rounded corners=7, thick, draw=black, minimum size=5mm](6) at (0,-1.5){$6$};
        \draw[->, very thick] (1) edge[dashed, "1"] (2) (2) edge (3) (3) edge (4) (4) edge[red] (5) (5) edge[red] (6) (6) edge[red] (1) (2) edge[bend right, red] (5) (5) edge[bend right, teal] (2);
    \end{tikzpicture} & \begin{tikzpicture}[every edge/.style = {draw, -Stealth}]]
        \node(l) at (1.5,-2.4){(i)};
        \node[circle, rounded corners=7, thick, draw=black, minimum size=5mm](1) at (0,0){$1$};
        \node[circle, rounded corners=7, thick, draw=black, minimum size=5mm](2) at (1.5,0){$2$};
        \node[circle, rounded corners=7, double, double distance = 1pt, thick, draw=black, minimum size=5mm](3) at (3,0){$3$};
        \draw[->, thick];
        \node[circle, rounded corners=7, thick, draw=black, minimum size=5mm](4) at (3,-1.5){$4$};
        \node[circle, rounded corners=7, thick, draw=black, minimum size=5mm](5) at (1.5,-1.5){$5$};
        \node[circle, rounded corners=7, thick, draw=black, minimum size=5mm](6) at (0,-1.5){$6$};
        \draw[->, very thick] (1) edge[dashed, "1"] (2) (2) edge (3) (3) edge[red] (4) (4) edge[red] (5) (5) edge[red] (6) (6) edge[red] (1) (2) edge[bend right, red] (5) (5) edge[bend right, teal] (2);
    \end{tikzpicture} \\
 \begin{tikzpicture}[every edge/.style = {draw, -Stealth}]]
        \node(l) at (1.5,-2.4){(j)};
        \node[circle, rounded corners=7, thick, draw=black, minimum size=5mm](1) at (0,0){$1$};
        \node[circle, rounded corners=7, double, double distance = 1pt, thick, draw=black, minimum size=5mm](2) at (1.5,0){$2$};
        \node[circle, rounded corners=7, thick, draw=black, minimum size=5mm](3) at (3,0){$3$};
        \draw[->, thick];
        \node[circle, rounded corners=7, thick, draw=black, minimum size=5mm](4) at (3,-1.5){$4$};
        \node[circle, rounded corners=7, thick, draw=black, minimum size=5mm](5) at (1.5,-1.5){$5$};
        \node[circle, rounded corners=7, thick, draw=black, minimum size=5mm](6) at (0,-1.5){$6$};
        \draw[->, very thick] (1) edge[dashed, "1"] (2) (2) edge[teal] (3) (3) edge[red] (4) (4) edge[red] (5) (5) edge[red] (6) (6) edge[red] (1) (2) edge[bend right, red] (5) (5) edge[bend right, teal] (2);
    \end{tikzpicture} & \begin{tikzpicture}[every edge/.style = {draw, -Stealth}]]
        \node(l) at (1.5,-2.4){(k)};
        \node[circle, rounded corners=7, thick, draw=black, minimum size=5mm](1) at (0,0){$1$};
        \node[circle, rounded corners=7, thick, draw=black, minimum size=5mm](2) at (1.5,0){$2$};
        \node[circle, rounded corners=7, double, double distance = 1pt, thick, draw=black, minimum size=5mm](3) at (3,0){$3$};
        \draw[->, thick];
        \node[circle, rounded corners=7, thick, draw=black, minimum size=5mm](4) at (3,-1.5){$4$};
        \node[circle, rounded corners=7, thick, draw=black, minimum size=5mm](5) at (1.5,-1.5){$5$};
        \node[circle, rounded corners=7, thick, draw=black, minimum size=5mm](6) at (0,-1.5){$6$};
        \draw[->, very thick] (1) edge[dashed, "1"] (2) (2) edge[dashed, "2"] (3) (3) edge[red] (4) (4) edge[red] (5) (5) edge[red] (6) (6) edge[red] (1) (2) edge[bend right, red] (5) (5) edge[bend right, teal] (2);
    \end{tikzpicture} & \begin{tikzpicture}[every edge/.style = {draw, -Stealth}]]
        \node(l) at (1.5,-2.4){(l)};
        \node[circle, rounded corners=7, thick, draw=black, minimum size=5mm](1) at (0,0){$1$};
        \node[circle, rounded corners=7, thick, draw=black, minimum size=5mm](2) at (1.5,0){$2$};
        \node[circle, rounded corners=7, thick, draw=black, minimum size=5mm](3) at (3,0){$3$};
        \draw[->, thick];
        \node[circle, rounded corners=7, double, double distance = 1pt, thick, draw=black, minimum size=5mm](4) at (3,-1.5){$4$};
        \node[circle, rounded corners=7, thick, draw=black, minimum size=5mm](5) at (1.5,-1.5){$5$};
        \node[circle, rounded corners=7, thick, draw=black, minimum size=5mm](6) at (0,-1.5){$6$};
        \draw[->, very thick] (1) edge[dashed, "1"] (2) (2) edge[dashed, "2"] (3) (3) edge[dashed, "3"] (4) (4) edge[red] (5) (5) edge[red] (6) (6) edge[red] (1) (2) edge[bend right, red] (5) (5) edge[bend right, teal] (2);
    \end{tikzpicture} \\
 \begin{tikzpicture}[every edge/.style = {draw, -Stealth}]]
        \node(l) at (1.5,-2.4){(m)};
        \node[circle, rounded corners=7, thick, draw=black, minimum size=5mm](1) at (0,0){$1$};
        \node[circle, rounded corners=7, thick, draw=black, minimum size=5mm](2) at (1.5,0){$2$};
        \node[circle, rounded corners=7, thick, draw=black, minimum size=5mm](3) at (3,0){$3$};
        \draw[->, thick];
        \node[circle, rounded corners=7, thick, draw=black, minimum size=5mm](4) at (3,-1.5){$4$};
        \node[circle, rounded corners=7, double, double distance = 1pt, thick, draw=black, minimum size=5mm](5) at (1.5,-1.5){$5$};
        \node[circle, rounded corners=7, thick, draw=black, minimum size=5mm](6) at (0,-1.5){$6$};
        \draw[->, very thick] (1) edge[dashed, "1"] (2) (2) edge[dashed, "2"] (3) (3) edge[dashed, "3"] (4) (4) edge[dashed, "4"] (5) (5) edge[red] (6) (6) edge[red] (1) (2) edge[bend right, red] (5) (5) edge[bend right, teal] (2);
    \end{tikzpicture} & \begin{tikzpicture}[every edge/.style = {draw, -Stealth}]]
        \node(l) at (1.5,-2.4){(n)};
        \node[circle, rounded corners=7, thick, draw=black, minimum size=5mm](1) at (0,0){$1$};
        \node[circle, rounded corners=7, double, double distance = 1pt,thick, draw=black, minimum size=5mm](2) at (1.5,0){$2$};
        \node[circle, rounded corners=7, thick, draw=black, minimum size=5mm](3) at (3,0){$3$};
        \draw[->, thick];
        \node[circle, rounded corners=7, thick, draw=black, minimum size=5mm](4) at (3,-1.5){$4$};
        \node[circle, rounded corners=7, thick, draw=black, minimum size=5mm](5) at (1.5,-1.5){$5$};
        \node[circle, rounded corners=7, thick, draw=black, minimum size=5mm](6) at (0,-1.5){$6$};
        \draw[->, very thick] (1) edge[dashed, "1"] (2) (2) edge[dashed, "2"] (3) (3) edge[dashed, "3"] (4) (4) edge[dashed, "4"] (5) (5) edge[red] (6) (6) edge[red] (1) (2) edge[bend right, red] (5) (5) edge[bend right, dashed, "5"] (2);
    \end{tikzpicture} &
 \begin{tikzpicture}[every edge/.style = {draw, -Stealth}]]
        \node(l) at (1.5,-2.4){(o)};
        \node[circle, rounded corners=7, thick, draw=black, minimum size=5mm](1) at (0,0){$1$};
        \node[circle, rounded corners=7, thick, draw=black, minimum size=5mm](2) at (1.5,0){$2$};
        \node[circle, rounded corners=7, thick, draw=black, minimum size=5mm](3) at (3,0){$3$};
        \draw[->, thick];
        \node[circle, rounded corners=7, thick, draw=black, minimum size=5mm](4) at (3,-1.5){$4$};
        \node[circle, rounded corners=7, double, double distance = 1pt,thick, draw=black, minimum size=5mm](5) at (1.5,-1.5){$5$};
        \node[circle, rounded corners=7, thick, draw=black, minimum size=5mm](6) at (0,-1.5){$6$};
        \draw[->, very thick] (1) edge[dashed, "1"] (2) (2) edge[dashed, "2"] (3) (3) edge[dashed, "3"] (4) (4) edge[dashed, "4"] (5) (5) edge[red] (6) (6) edge[red] (1) (2) edge[bend right, dashed] node[left]{6} (5) (5) edge[bend right, dashed, "5"] (2);
    \end{tikzpicture} \\
 \begin{tikzpicture}[every edge/.style = {draw, -Stealth}]]
        \node(l) at (1.5,-2.4){(p)};
        \node[circle, rounded corners=7, thick, draw=black, minimum size=5mm](1) at (0,0){$1$};
        \node[circle, rounded corners=7, thick, draw=black, minimum size=5mm](2) at (1.5,0){$2$};
        \node[circle, rounded corners=7, thick, draw=black, minimum size=5mm](3) at (3,0){$3$};
        \draw[->, thick];
        \node[circle, rounded corners=7, thick, draw=black, minimum size=5mm](4) at (3,-1.5){$4$};
        \node[circle, rounded corners=7, thick, draw=black, minimum size=5mm](5) at (1.5,-1.5){$5$};
        \node[circle, rounded corners=7, double, double distance = 1pt,thick, draw=black, minimum size=5mm](6) at (0,-1.5){$6$};
        \draw[->, very thick] (1) edge[dashed, "1"] (2) (2) edge[dashed, "2"] (3) (3) edge[dashed, "3"] (4) (4) edge[dashed, "4"] (5) (5) edge[dashed, "7"] (6) (6) edge[red] (1) (2) edge[bend right, dashed] node[left]{6} (5) (5) edge[bend right, dashed, "5"] (2);
    \end{tikzpicture} &
 \begin{tikzpicture}[every edge/.style = {draw, -Stealth}]]
        \node(l) at (1.5,-2.4){(q)};
        \node[circle, rounded corners=7, double, double distance = 1pt,thick, draw=black, minimum size=5mm](1) at (0,0){$1$};
        \node[circle, rounded corners=7, thick, draw=black, minimum size=5mm](2) at (1.5,0){$2$};
        \node[circle, rounded corners=7, thick, draw=black, minimum size=5mm](3) at (3,0){$3$};
        \draw[->, thick];
        \node[circle, rounded corners=7, thick, draw=black, minimum size=5mm](4) at (3,-1.5){$4$};
        \node[circle, rounded corners=7, thick, draw=black, minimum size=5mm](5) at (1.5,-1.5){$5$};
        \node[circle, rounded corners=7, thick, draw=black, minimum size=5mm](6) at (0,-1.5){$6$};
        \draw[->, very thick] (1) edge[dashed, "1"] (2) (2) edge[dashed, "2"] (3) (3) edge[dashed, "3"] (4) (4) edge[dashed, "4"] (5) (5) edge[dashed, "7"] (6) (6) edge[dashed, "8"] (1) (2) edge[bend right, dashed] node[left]{6} (5) (5) edge[bend right, dashed, "5"] (2);
    \end{tikzpicture} &
\end{tabular}
\caption{Sample execution of \textsc{Space-Efficient-Hierholzer}. For clarity, \textsc{Dashed} edges are labelled in the order in which they are written. The current vertex is represented by double borders.}
    \label{fig:exec}
\end{figure}

When we reach a vertex $v$ without \textsc{Black} incoming edges, we backtrack from it via an \emph{outgoing} edge according to the following rule:
\begin{enumerate}
    \item If $v$ has an outgoing \textsc{Green} edge, backtrack along it;
    \item otherwise, if $v \ne v_0$, backtrack via its (unique) \textsc{Red} outgoing edge;
    \item otherwise, if $v = v_0$, terminate the algorithm.
\end{enumerate}
Each time we backtrack across an edge (irrespective of whether it was \textsc{Green} or \textsc{Red} before), we mark it \textsc{Dashed} and we output it as the next edge of the Eulerian cycle. 
In this way, the Eulerian cycle is produced sequentially, in the correct order. 
Figure~\ref{fig:exec} shows a sample execution of the algorithm. 
Notice that, at step (m) of the figure, the traversal proceeds to vertex $2$ instead of vertex $6$; this is due to the backtracking rules, which prioritise \textsc{Green} edges over \textsc{Red} ones. 
Indeed, if the traversal proceeded to vertex $6$ instead, it would get stuck at vertex $1$, and the two edges between vertices $2$ and $5$ would never be counted as part of the computed cycle. 
The backtracking rules ensure that such a scenario does not happen.

\subsection{Space-Efficient Implementation.}

\begin{algorithm}[tbh]
\caption{\textsc{Space-Efficient-Hierholzer}}\label{alg:euler-tour}
\hspace*{\algorithmicindent} \textbf{Input:} An Eulerian directed graph $G$ with $n$ vertices and $m$ edges.\\
\hspace*{\algorithmicindent} \textbf{Output:} An Eulerian cycle, incrementally written.
\begin{algorithmic}[1]
\State $\texttt{next}[1..n],~\texttt{visited}[1..n],~\texttt{skipped}[1..n],~B[1..n] \gets [0]^n$
\State $c \gets 0$
\State $v_0 \gets k \in [n]$ \Comment{Pick an arbitrary starting vertex.}
\State $u \gets v_0$ \Comment{Initially, the current vertex is the starting one.}
\State $\texttt{visited}[u] \gets 1$
\While{$c < m$} \Comment{We loop as long as there are still edges to write.}
    \State $\texttt{next}[u] \gets \texttt{next}[u] + 1$
    \State $i \gets \texttt{next}[u]$
    \If{$i \leq d^-(u)$} \Comment{Reverse traversal phase.}
        \State $v \gets \Gamma^-(u, i)$
        \If{$\texttt{visited}[v] = 0$}
            \If{$v \ne v_0$}
                \State $B[v] \gets u$
            \EndIf
            \State $\texttt{visited}[v] \gets 1$
        \EndIf
        \State $u \gets v$
    \Else \Comment{Backtracking phase.}
        \State $i \gets \texttt{next}[u] - d^-(u)$
        \If{$\Gamma^+(u, i) = B[u]$ \textbf{and} $\texttt{skipped}[u] = 0$}
            \State $\texttt{skipped}[u] \gets 1$
            \State $\texttt{next}[u] \gets \texttt{next}[u] + 1$
            \State $i \gets i+1$
        \EndIf
        \If{$i > d^+(u)$}
            \State $v \gets B[u]$
        \Else
            \State $v \gets \Gamma^+(u, i)$
        \EndIf
        \State \textsc{Write}$(uv)$
        \State $c \gets c + 1$
        \State $u \gets v$
    \EndIf
\EndWhile
\end{algorithmic}
\end{algorithm}

We now describe our space-efficient implementation of the conceptual algorithm, see Algorithm~\ref{alg:euler-tour}.
An interactive visualization of the algorithm is available at \url{https://www.wild-inter.net/publications/hierholzer}.
For each vertex $u$, the algorithm uses a counter $\texttt{next}[u]$ to keep track of how many of its incoming and outgoing edges have been traversed. 
The bit vector \texttt{visited} stores whether a vertex has been seen before, and \texttt{skipped} ensures that the special backtracking edge (the \textsc{Red} outgoing edge) is not used until all other outgoing edges have been used. 

When a vertex $v$ is visited for the first time (via an outgoing edge $vu$), the algorithm records $u$ as $v$'s predecessor by setting $B[v] \gets u$; this is equivalent to marking $uv$ \textsc{Red}. 
Backtracking begins when a vertex $u$ has no unexplored incoming (\textsc{Black}) edges, which occurs when $\texttt{next}[u] > d^-(u)$ (i.e.\ when $\texttt{next}[u]$ is out of the bounds of $d^-(u)$). 
At that point, the algorithm starts following outgoing edges from $u$. To ensure the \textsc{Red} edge is used last, it checks whether the next candidate edge (indexed by $\texttt{next}[u] - d^-(u)$) leads to $B[u]$. If so, and if it has not already been skipped, the algorithm skips this edge. 
This guarantees that all non-\textsc{Red} (i.e.\ \textsc{Green}) edges are backtracked over first.%
\footnote{%
	Since the graph may contain parallel edges, the target vertex $v=B[u]$ could appear multiple times in the multiset $\{\Gamma^+(u,i):i \in [d^+(u)]\}$. Hence we use \texttt{skipped}-array to only skip the first edge $uv$.
}
Finally, after all other outgoing edges have been used (i.e.\ when $\texttt{next}[u]>d^+(u)+d^-(u)=d(u)$), the \textsc{Red} edge becomes the only option and is followed. 
As the algorithm backtracks from vertices, each traversed edge is immediately written to the output.

\section{Analysis.}
\label{s:analysis}

We now analyse the correctness and complexity of the algorithm. In Subsection \ref{ss:correctness}, we mainly use the high-level abstraction involving edge marks (\textsc{Black}, \textsc{Red}, \textsc{Green}, and \textsc{Dashed}) to prove that the algorithm outputs a valid Eulerian cycle. In Subsection \ref{ss:complexity}, we analyse the running time and space usage of the algorithm by referring directly to the pseudocode in Algorithm~\ref{alg:euler-tour}.

\subsection{Proof of Correctness.}
\label{ss:correctness}
For a vertex $v$ and a mark $X$, let $d^+(v, X)$ denote the number of outgoing edges from $v$ marked $X$, and define $d^-(v, X)$ analogously as the number of incoming edges to $v$ marked $X$. 

\begin{definition}[Needle vertex]
\label{d:needle-vertex}
    Let $u$ denote the current vertex of the traversal. A vertex $w$ is a \emph{needle vertex} if and only if:
    \begin{itemize}
        \item $w = u$ and $d^+(w, \textsc{Black}) = d^-(w, \textsc{Black})$; or 
        \item $w \neq u$ and $d^+(w, \textsc{Black}) = d^-(w, \textsc{Black}) + 1$.
    \end{itemize}
\end{definition}

Our first claim and corollary establish a crucial invariant of \textsc{Space-Efficient-Hierholzer} by showing that there is exactly one needle vertex throughout the execution of the algorithm. Our follow-up claim states that backtracking occurs from $u$ only when $u$ is a needle vertex.

\begin{lemma}
    \label{c:2-statements}
    At every step of the execution of \textsc{Space-Efficient-Hierholzer}, exactly one of the following statements holds:
    \begin{enumerate}
        \item for all $v \in V$, $d^+(v, \textsc{Black}) = d^-(v, \textsc{Black})$;
        \item there are two vertices $a$ and $b$ s.t.: 
        \begin{itemize}
            \item $d^+(a, \textsc{Black}) = d^-(a, \textsc{Black})+1$ and $d^+(b, \textsc{Black})+1 = d^-(b, \textsc{Black})$,
            \item for all $v \in V \setminus \{a,b\}$, $d^+(v, \textsc{Black}) = d^-(v, \textsc{Black})$, and
            \item $b$ is the current vertex.
        \end{itemize}
    \end{enumerate}
\end{lemma}

\begin{proof}
    We prove this claim by induction on the number of traversed edges. The base case is trivial, as no edges have been traversed. Hence, for all vertices $v \in V$, we have $d^+(v, \textsc{Black}) = d^-(v, \textsc{Black})$, and the invariant is satisfied.

    Now, suppose the invariant holds after $k$ edge traversals. We consider the $(k+1)$st edge traversal, which must either be a forward step over a \textsc{Black} incoming edge, or a backtracking step over a \textsc{Red} or \textsc{Green} outgoing edge, from some vertex $u$.

    \paragraph{Case 1.}
        \textit{Statement 1 holds after $k$ edge traversals.} Let $u$ be the current vertex. By the induction hypothesis, we have $d^+(u, \textsc{Black}) = d^-(u, \textsc{Black})$. If $d^+(u, \textsc{Black}) = 0$, the traversal backtracks via an outgoing \textsc{Red} or \textsc{Green} edge, meaning that neither $d^+ (v,\textsc{Black})$ nor $d^-(v,\textsc{Black})$ of any vertex $v \in V$ is affected, preserving the invariant. Otherwise, if $d^+(u, \textsc{Black}) > 0$, the traversal walks forward via one incoming \textsc{Black} edge and marks it either \textsc{Red} or \textsc{Green}. Let $w$ be the new current vertex. Vertex $u$ loses one \textsc{Black} incoming edge, and $w$ loses one \textsc{Black} outgoing edge; thus, we have $d^+(u, \textsc{Black}) = d^-(u, \textsc{Black})+1$ and $d^+(w, \textsc{Black})+1 = d^-(w, \textsc{Black})$, the measure of \textsc{Black} edges is not affected for vertices in $V \setminus \{u,w\}$, and $w$ becomes the new current vertex, satisfying the invariant.
    \paragraph{Case 2.}
        \textit{Statement 2 holds after $k$ edge traversals.} We then have $u=b$. Since $d^+(b, \textsc{Black})+1 = d^-(b, \textsc{Black})$, the traversal then proceeds by walking forward towards some vertex $x$. Consider two possible subcases.
            \paragraph{Subcase 2.1.}
            \textit{$x=a$.} By entering $a$ from $b$, $b$ loses one \textsc{Black} incoming edge, $a$ loses one \textsc{Black} outgoing edge, and thus we have $d^+(v, \textsc{Black}) = d^-(v, \textsc{Black})$ for all $v \in V$.
            \paragraph{Subcase 2.2.} \textit{$x \neq a.$} Before entering $x$ from $b$, we have $d^+(x, \textsc{Black}) = d^-(x, \textsc{Black})$; upon entering it, we have $d^+(x, \textsc{Black})+1 = d^-(x, \textsc{Black})$, $d^+(b, \textsc{Black}) = d^-(b, \textsc{Black})$, $d^+(a, \textsc{Black}) = d^-(a, \textsc{Black})+1$, and $x$ is the new current vertex.
\end{proof}

The following corollary immediately follows from Lemma~\ref{c:2-statements} and Definition~\ref{d:needle-vertex}.

\begin{corollary}
\label{cor:one-needle}
At every step of the execution of \textsc{Space-Efficient-Hierholzer}, there is exactly one needle vertex.
\end{corollary}

\begin{lemma}
\label{c:backtrack-needle}
    At every step of the execution of \textsc{Space-Efficient-Hierholzer}, backtracking occurs from $u$ only when $u$ is a needle vertex.
\end{lemma}
\begin{proof}
    It is established in Corollary~\ref {cor:one-needle} that exactly one needle vertex exists in the graph throughout the execution of the algorithm; thus, consider $w$ to be the needle vertex, and for the sake of contradiction, suppose that the algorithm backtracks from $u \neq w$. By Lemma~\ref{c:2-statements}, we arrive at two distinct cases.
    \paragraph{Case 1.}
        \textit{Statement 1 of Lemma~\ref{c:2-statements} holds.} Before backtracking, we have $d^+(v, \textsc{Black}) = d^-(v, \textsc{Black})$ for all $v \in V$. Since there is exactly one needle vertex, then it must be the current vertex $u$ by Definition~\ref{d:needle-vertex}, which is in contradiction with the assumption that the current vertex is not the needle vertex.
    \paragraph{Case 2.}
        \textit{Statement 2 of Lemma~\ref{c:2-statements} holds.} We then have $d^+(u, \textsc{Black})+1 = d^-(u, \textsc{Black})$, implying that prior to backtracking from $u$, there existed a \textsc{Black} incoming edge incident on $u$, which is not possible by the definition of the algorithm.

        This completes the proof.
\end{proof}

\begin{lemma}
\label{lem:only-backtrack}
    The needle vertex changes only upon backtracking.
\end{lemma}

\begin{proof}
    Towards a contradiction, suppose otherwise; that is, assume that a forward step into $v$ makes $v$ the needle vertex. Before this forward step, we had $d^+(v, \textsc{Black}) = d^-(v, \textsc{Black})$. After the step, however, $d^+(v, \textsc{Black}) + 1 = d^-(v, \textsc{Black})$. Since $v$ is now the current vertex and $d^+(v, \textsc{Black}) \neq d^-(v, \textsc{Black})$, this contradicts Definition~\ref{d:needle-vertex}.
\end{proof}

\begin{remark}
    \label{r:first-backtrack-v0}
    A consequence of Lemmata~\ref{c:backtrack-needle} and~\ref{lem:only-backtrack} is that backtracking occurs for the first time at $v_0$ (the first needle vertex), and then makes a vertex incident on $v_0$ the new needle vertex. Observe that this process indeed forms a trail of \textsc{Dashed} edges in the graph.
\end{remark}

We now show that this trail is closed at termination, i.e.\ forms a cycle where $v_0$ is included. To do so, it suffices to show that the algorithm terminates and that $v_0$ is the last current vertex. 

To prove termination, we fix the potential function $c:M \to \mathbb{N} \cup\{0\}$ with $M=\{\textsc{Black},~ \textsc{Red},\allowbreak~\textsc{Green},~\textsc{Dashed}\}$, $c(\textsc{Black})=3$, $c(\textsc{Red})=2$, $c(\textsc{Green})=2$, and $c(\textsc{Dashed})=0$.

\begin{lemma}
\label{c:termination}
    The algorithm \textsc{Space-Efficient-Hierholzer} terminates.
\end{lemma}
\begin{proof}
    Let $\Phi$ be a potential function defined as $\Phi(G)=\sum_{e \in E}c(m(e))$ for a graph $G=(V,E)$, where $m(e)$ denotes the mark of edge $e$; in other words, $\Phi$ is equal to the total sum of the costs of all edges in $G$. On an input graph $G$, we show that at each step of the traversal, the measure $\Phi(G)$ strictly decreases.

    Initially, all edges are \textsc{Black}, hence $\Phi(G)=3|E|$. Observe that each edge changes its mark over time in a fixed and monotonic order: $\textsc{Black} \to \{\textsc{Red},\textsc{Green}\} \to \textsc{Dashed}$. Since each step forward or backwards in the traversal alters the mark of an edge, and no \textsc{Red} or \textsc{Green} edge is remarked \textsc{Green} or \textsc{Red} respectively, the measure strictly decreases at each step, and the algorithm eventually terminates as $\Phi(G) \geq 0$.
\end{proof}

\begin{lemma}
\label{c:last-current}
    Upon termination of \textsc{Space-Efficient-Hierholzer}, the \textsc{Dashed} edges form a closed trail which includes $v_0$ as an endpoint.
\end{lemma}
\begin{proof}
    By Remark~\ref{r:first-backtrack-v0}, the algorithm first backtracks from $v_0$. Observe that, by the definition of an Eulerian graph, entering a non-starting vertex via an outgoing \textsc{Black} edge (which is then marked \textsc{Red} or \textsc{Green}) guarantees that there is a corresponding incoming \textsc{Black} edge that will eventually be used to exit the vertex and marked accordingly. Similarly, entering a non-starting vertex via a \textsc{Red} or \textsc{Green} outgoing edge (which is then marked \textsc{Dashed}) ensures that there is a corresponding incoming \textsc{Red} or \textsc{Green} edge that will eventually be used to exit the vertex and also be marked \textsc{Dashed}. Therefore, at every step of the algorithm, if $u$ is the current vertex, the following conditions hold:
    \begin{itemize}
        \item $d^+(v, \textsc{Red}) + d^+(v, \textsc{Green}) = d^-(v, \textsc{Red}) + d^-(v, \textsc{Green})$ for all $v \in V \setminus \{v_0, u\}$, and
        \item $d^+(u, \textsc{Red}) + d^+(u, \textsc{Green}) = d^-(u, \textsc{Red}) + d^-(u, \textsc{Green}) + 1$ if $u \neq v_0$.
    \end{itemize}
    Towards a contradiction, suppose now that the last current vertex is $v \neq v_0$ upon termination. This implies that there was no \textsc{Red} or \textsc{Green} outgoing edge to backtrack through from $v$, hence the termination. However, if $v$ were entered via a \textsc{Black} incoming edge, then, by the observation above, there would have been a \textsc{Black} incoming edge to resume the traversal, a contradiction. If, analogously, $v$ were entered via a \textsc{Green} or \textsc{Red} edge, then there would be at least one \textsc{Red} or \textsc{Green} outgoing edge to backtrack through, a contradiction. Therefore, the last current vertex cannot be in $V\setminus \{v_0\}$, and can only be $v_0$ since the algorithm terminates by Lemma~\ref{c:termination}.
\end{proof}

We now need to show that the closed trail produced by the algorithm includes all edges in the graph; it suffices to show that all edges are \textsc{Dashed} upon termination. First, we show that there are no \textsc{Black} edges at termination, indicating that all edges have been traversed at least once. The following structural property of Eulerian graphs comes in handy.

\begin{lemma}
    \label{o:remove-trail}
    Let $G$ be an Eulerian graph, and let $T$ be some arbitrary closed trail in $G=(V,E)$. Then, $G'=(V,E \setminus T)$ consists of multiple (possibly none) vertex-disjoint closed trails. 
\end{lemma}
\begin{proof}
    Initially, for all $v \in V$, we have $d^+(v)=d^-(v)$. Clearly, removing the trail $T$ preserves that equality. Let $H$ be some arbitrary weakly\footnote{A set of vertices is considered \emph{weakly connected} if there exists a path of edges from every pair of vertices within that set when edge directions are ignored.} connected component of $G'$; then, for all vertices $v_H \in H$, we have $d^+(v_H)=d^-(v_H)$. Hence, $H$ is  Eulerian and its edges form a closed trail.
\end{proof}


\begin{lemma}
    \label{c:no-blacks}
    Upon termination of \textsc{Space-Efficient-Hierholzer}, no \textsc{Black} edges remain.
\end{lemma}
\begin{proof}
    Let $T$ be the trail of \textsc{Dashed} edges computed by \textsc{Space-Efficient-Hierholzer} by Lemma~\ref{c:last-current}; for the sake of contradiction, we split the proof into two distinct assumptions.
    \paragraph{Case 1.}
    \textit{A \textsc{Black} edge is incident on a vertex $v \in T$.} If $v$ has a \textsc{Black} incoming edge, this implies that the algorithm backtracked from or terminated at $v$ despite there being a \textsc{Black} incoming edge to traverse, which is a contradiction. If, on the other hand, the \textsc{Black} edge is outgoing from $v$, then by Lemma~\ref{c:2-statements}, there must also be an incoming \textsc{Black} edge at $v$. Indeed, Statement 1 of the lemma applies, since the current vertex cannot possess incoming \textsc{Black} edges without violating the algorithm’s definition (i.e.\ backtracking would have occurred instead of termination). Therefore, the previous subcase applies.

    \paragraph{Case 2.}
    \textit{No \textsc{Black} edge is incident on a vertex $v \in T$.} By Lemma~\ref{o:remove-trail}, we can remove $T$ and obtain multiple Eulerian connected components, each containing a closed trail. Let $H_1,H_2, \dots, H_k$ be the $k$ vertex-disjoint strongly connected components obtained. Now, define $S_i=V_T \cap H_i$, where $V_T$ is the set of endpoints in $T$; in other words, $S_i$ is the set of all vertices that are part of both the trail and the $i$th connected component. We claim that for each $S_i$ for $i \in [k]$, there is at least one vertex $w \in S_i$ such that $d^+(w, \textsc{Green}) \geq 1$, and $d^+(w, \textsc{Red}) = 0$.

    First, observe that the \textsc{Red} edges marked by \textsc{Space-Efficient-Hierholzer} form a tree rooted at $v_0$, the starting vertex; hence $d^+(v, \textsc{Red}) \leq 1$ for any vertex $v$. Pick an arbitrary set $S_i$: using the invariant from the proof of Lemma~\ref{c:last-current}, for each $w \in S_i$, we have $d^+(w, \textsc{Red})+d^+(w, \textsc{Green}) = d^-(w, \textsc{Red})+d^-(w, \textsc{Green}) > 0$. If $v_0 \in S_i$, then we are done since $d^+(v_0, \textsc{Red})=0$ and thus $d^+(v_0, \textsc{Green}) \geq 1$. Now, towards a contradiction, suppose that we have $d^+(w, \textsc{Red})=1$ for all $w \in S_i$ (recall that this measure cannot be greater than $1$), and $v_0 \not\in S_i$. Since $H_i$ does not share vertices with other components in $\biguplus_{j \in [k]\setminus\{i\}}H_j$, there is no path from a vertex in $H_i \setminus S_i$ to $v_0$ that does not include an outgoing edge incident on a vertex in $S_i$. Thus, the assumption entails that the starting vertex (i.e.\ the root of the tree of \textsc{Red} edges) $v_0 \in H_i \setminus S_i$, which contradicts our assumption that $v_0 \in T$ since $T \cap (H_i \setminus S_i) = \emptyset$.

        Now, as we have shown that a vertex $w \in T$ such that $d^+(w, \textsc{Green}) \geq 1$ and $d^+(w, \textsc{Red}) = 0$ exists, we can pick one arbitrarily and consider two distinct cases. If $w = v_0$, then the algorithm terminated at $v_0$ despite $d^+(v_0, \textsc{Green}) \geq 1$, which is impossible by the way the algorithm is defined. Finally, if $w \neq v_0$, then the algorithm backtracked via a \textsc{Red} outgoing edge (later marked \textsc{Dashed}) despite there being a \textsc{Green} outgoing edge to backtrack through, which contradicts the definition of the algorithm.
        
    In either case, we reach a contradiction; therefore, no \textsc{Black} edges remain in the graph.
\end{proof}

The following corollary strictly follows from Case 2 of the proof of Lemma~\ref{c:no-blacks}.

\begin{corollary}
    \label{cor:all-dashed}
    Upon termination of \textsc{Space-Efficient-Hierholzer}, all edges are \textsc{Dashed}.
\end{corollary}
\begin{proof}
    For the sake of contradiction, suppose that, at termination, the graph contains \textsc{Red} or \textsc{Green} edges; recall that it has been established in Lemma~\ref{c:no-blacks} that the graph cannot contain \textsc{Black} edges. Let $T$ be the trail of \textsc{Dashed} edges produced by the algorithm (Lemma~\ref{c:last-current}). Then, by the argument of Case 2 of the proof of Lemma~\ref{c:no-blacks}, there must exist some vertex $w \in T$ such that $d^+(w, \textsc{Green}) \geq 1$ and $d^+(w, \textsc{Red}) = 0$, which leads to contradicting the definition of the algorithm. Therefore, all edges are \textsc{Dashed} at termination.
\end{proof}

\begin{theorem}
    The algorithm \textsc{Space-Efficient-Hierholzer} correctly computes an Eulerian cycle.
\end{theorem}
\begin{proof}
    By Lemma~\ref{c:termination}, the algorithm terminates after a finite number of steps. Lemma~\ref{c:last-current} establishes that, upon termination, the set of \textsc{Dashed} edges forms a closed trail containing the starting vertex $v_0$. Moreover, Lemma~\ref{lem:only-backtrack} and Lemma~\ref{c:backtrack-needle} together ensure that a new \textsc{Dashed} edge is created exactly when the traversal backtracks from the current needle vertex, and that the next needle vertex is incident to the previously \textsc{Dashed} one, thus producing the trail in a correct order. Finally, Corollary~\ref{cor:all-dashed} implies that all edges are marked \textsc{Dashed} upon termination, which ensures that the closed trail includes every edge exactly once.
\end{proof}

This completes the proof of correctness.
\subsection{Space and Time Complexity.}
\label{ss:complexity}
Recall that we assume computing the $i$th in- or out-neighbour of a vertex, or its in- or out-degree, takes $\mathrm{O}(1)$ time.

\begin{theorem}
    The algorithm \textsc{Space-Efficient-Hierholzer} (Algorithm~\ref{alg:euler-tour}) terminates in time $\mathrm{O}(m)$, where $m$ is the number of edges in the input graph.
\end{theorem}
\begin{proof}
    Termination follows from Lemma~\ref{c:termination}. Initialising the arrays $\texttt{next}$, $\texttt{visited}$, $\texttt{skipped}$ and $B$ takes $\mathrm{O}(n)$ time. We claim that the \textbf{while} loop (Lines 6-28) iterates $2m$ times. The loop terminates when $c \ge m$; and during an iteration, $c$ is incremented if and only if $\texttt{next}[u] > d^-(u)$. Since the latter is incremented at each iteration of the loop, the \textbf{if} statement (Line 9) is satisfied exactly $\sum_{v \in V}d^-(v)=m$ times. Thus, the \textbf{else} block (Line 16) is executed $m$ times.

    The overall complexity is $\mathrm{O}(n+m) \in \mathrm{O}(m)$ since, in any Eulerian graph of $2$ or more vertices, $n \leq m$.
\end{proof}

\begin{theorem}
    The algorithm \textsc{Space-Efficient-Hierholzer} (Algorithm~\ref{alg:euler-tour}) uses $\mathrm{O}(n\lg m)$ bits of working space.
\end{theorem}

\begin{proof}
    The arrays $\texttt{visited}$ and $\texttt{skipped}$ take $\mathrm{O}(n)$ bits, as these represent bit vectors. Values in the arrays $\texttt{next}$ and $B$ are bounded by $m$, and thus the arrays take $\mathrm{O}(n \lg m)$ bits. 
\end{proof}

\section{Conclusion.}
\label{s:conclusion}
We presented a simple and memory-efficient algorithm for computing Eulerian cycles in directed graphs. Unlike classical implementations of Hierholzer’s algorithm, which require storing all edges or large stacks, our method uses only $\mathrm{O}(n\lg m)$ bits of working space by storing a small amount of data per vertex rather than per edge. Despite the reduced memory usage, our algorithm still runs in linear time and produces the Eulerian cycle incrementally, without needing to reverse it at the end. Our approach is especially useful for dense graphs, where the number of edges is much larger than the number of vertices, and in settings where memory is limited. It is also well-suited for environments such as graph rewriting systems that do not support stacks or pointer-based data structures.

Our approach works on directed graphs where the direction immediately tells us if we are seeing an edge the second time (green edge), so we can now traverse and append it to the output.
For undirected graphs, \cite{grigorescu2026space} describes a solution not based on edge-centric DFS.
Whether an extension of our traversal-based approach to undirected graphs is possible
remains open for future work.

%

\subsection*{Acknowledgements.}

We would like to thank Viktor Zamaraev and Nikhil Mande for the initial discussions on the topic.

\bibliographystyle{alphaurl}
\bibliography{bib}

\end{document}